\documentclass[a4paper,UKenglish,cleveref, autoref, thm-restate]{lipics-v2021}

\hideLIPIcs  

\usepackage[linesnumbered,ruled,vlined]{algorithm2e}
\usepackage{float}
\usepackage{booktabs}
\newcommand{\Z}{\mathbb Z}
\newcommand{\Vmin}{V_{\mathrm{Min}}}
\newcommand{\Vmax}{V_{\mathrm{Max}}}

\newcommand{\val}{\mathrm{val}}

\newcommand{\MP}{\mathrm{MP}}

\newcommand{\re}[1]{\xrightarrow{#1}}

\newcommand{\esc}{\mathrm{esc}}

\newcommand{\supsum}{\sup \! \Sigma}
\newcommand{\infsum}{\inf \! \Sigma}

\newcommand{\pred}{\mathrm{pred}}

\title{A symmetric recursive algorithm for mean-payoff games} 

\author{Pierre Ohlmann}{CNRS, Laboratoire d'Informatique et des Systèmes (LIS), Marseille, France}{pierre.ohlmann@lis-lab.fr}{}{}

\authorrunning{P. Ohlmann} 

\Copyright{Jane Open Access and Joan R. Public} 



\category{} 

\relatedversion{} 

\acknowledgements{We are grateful to  Pierre Gaillard, Antonio Casares and Aditya Prakash for discussions on this algorithm.}

\nolinenumbers 

\EventEditors{John Q. Open and Joan R. Access}
\EventNoEds{2}
\EventLongTitle{42nd Conference on Very Important Topics (CVIT 2016)}
\EventShortTitle{CVIT 2016}
\EventAcronym{CVIT}
\EventYear{2016}
\EventDate{December 24--27, 2016}
\EventLocation{Little Whinging, United Kingdom}
\EventLogo{}
\SeriesVolume{42}
\ArticleNo{23}

\begin{document}

\maketitle

\begin{abstract}
We propose a new deterministic symmetric recursive algorithm for solving mean-payoff games.
\end{abstract}

\section{Introduction}

In a mean-payoff game, Min and Max alternate in moving a token along the integer-weighted edges of a sinkless directed graph.
The play goes on for an infinite duration, and its value is defined to be the average of the weights on the long run.
These games were introduced by Ehrenfeucht and Mycielski~\cite{EM73} who proved their positional determinacy: to play optimally, players do not need memory, and can always make the same choice from each given vertex.
They were later studied in more depth by Gurvich, Karzanov and Khachyian~\cite{GKK88}, and then Zwick and Paterson~\cite{ZP96} who popularised the problem of determining which vertices have positive value, observed that it belongs to $\bf{NP} \cap \bf{coNP}$ (this is a simple consequence of positional determinacy), gave a pseudopolynomial algorithm, and posed the question of solving them in polynomial time.

As of today, a number of algorithms have been proposed, most prominent of which are the value iteration algorithm of Brim et al.~\cite{BCDGR11}, the GKK algorithm~\cite{GKK88}, and the strategy improvement algorithms of Björklund and Vorobiov~\cite{BV07} and of Schewe~\cite{Schewe08}; for a unified presentation of many algorithms, see also~\cite{CCO25}.
In terms of runtime bounds, if $n,m$ and $W$ respectively denote the number of vertices, number of edges, and maximal absolute value of a weight, the best known pseudopolynomial upper bound is $O(nmW \log(W))$ from~\cite{BCDGR11}, and the best known combinatorial bounds are $O(m2^{\frac n 2} \log(W))$ for the GKK algorithm~\cite{DKZ19,OhlmannGKK}\footnote{This was first established by Dorfman, Kaplan and Zwick for a involved variant of the value iteration algorithm. We believe that the resulting algorithm is in fact similar in spirit to the (conceptually simpler) GKK algorithm, and proved the same bound in~\cite{OhlmannGKK}.}, while the strategy improvement of~\cite{BV07}, akin to random facet for linear programs, runs in randomised subexponential time~$O(2^{O(\sqrt{n \log n})}\log(W))$.
A notable recent advancement from Loff and Skomra is an algorithm which runs in polynomial time over generic mean-payoff games, in the sense of smooth analysis~\cite{LS24}.
However, despite considerable efforts from the community spanning over several decades, no deterministic subexponential time algorithm is known for solving mean-payoff games.


Our contribution is a new deterministic algorithm solving mean-payoff games.
It is completely symmetric (i.e.~the two players are treated in the same way); from the algorithms discussed above, this property is shared only by the GKK algorithm.
It is naturally presented recursively, which is another novelty.
Another significant difference between our algorithm and every known mean-payoff algorithm, with the exception of the one of Zwick and Paterson~\cite{ZP96}, is that we do not compute energy values (see preliminaries for a definition).
Though we rely on potential reductions, which were introduced by Gurvich, Karzanov and Khachyian~\cite{GKK88} and are at the basis of the unifying framework proposed by Cadilhac, Casares and the author~\cite{CCO25}, our algorithm does not fit in this framework. 
We believe that our algorithm is a candidate for a subexponential runtime, but we leave such an upper bound to future work.

After a preliminary section (Section~\ref{sec:prelim}), we describe the algorithm (Section~\ref{sec:algorithmpresentation}), prove its correctness (Section~\ref{sec:proofs}), give a pseudocode (Section~\ref{sec:pseudocode}), propose optimisations and variants (Section~\ref{sec:optimisations-and-variant}) and conclude with a brief discussion (Section~\ref{sec:futurework}).

\section{Preliminaries and background}\label{sec:prelim}

We now give some background regarding mean-payoff games.
Our aim here is not to give only the strictly necessary definitions to be able to read the rest of the paper, but also we want to provide a bit of background and lay out a few intuitions to make the reading easier.
When it is not too much of a hassle, we try to avoid confusing terminology such as positive or non-negative, and instead just write $>0$ or $\geq 0$.

A game is a sinkless\footnote{A sink is a vertex with no outgoing edge.} directed graph $G=(V,E,w)$ whose edges are weighted by $w:E \to \Z$, and whose vertices are partitioned into $\Vmin$ and $\Vmax$ which we call Min-vertices and Max-vertices.
Intuitively, when we are at a Min-vertex, Min chooses which outgoing edge to follow, and likewise for Max.
A subgame is a sinkless subgraph.
A trap for Min is a subgame $G'$ such that all Min edges from $G'$ are towards $G'$, and similarly for Max.

As in the introduction, we write $n=|V|$, $m=|E|$ and $W=\max_{vv' \in E} |w(vv')|$.
We sometimes write edges and their weights as $v \re w v'$, and (finite and infinite) paths as $v_0 \re {w_0} v_1 \re{w_1} \dots$.
The sequence of partial sums of prefixes of a path, namely $0,w_0, w_0+w_1, \dots$, is called its profile.

A Min-strategy\footnote{We only consider positional strategies in this paper as they are sufficient for the games under scrutiny.} is comprised of a choice, for each Min vertex, of one of its outgoing edges.
A path is consistent with a Min-strategy if whenever a Min-vertex is visited, the next edge is the one specified by the strategy.
We will write $\pi \models \sigma$ to say that path $\pi$ is consistent with strategy $\sigma$.
Similarly for Max-strategies.
A valuation is a map from infinite paths to $\overline {\mathbb R}$; the three valuations which are relevant to this paper are the following:
\[
    \begin{array}{lcl}
        \MP(v_0 \re {w_0} v_1 \re{w_1} \dots) &=& \limsup_k \frac 1 k \sum_{i=0}^{k-1} w_i,\\
        \supsum^X(v_0 \re {w_0} v_1 \re{w_1} \dots) &=& \sup_{k\leq n_X} \sum_{i=0}^{k-1} w_i, \\
        \infsum^X(v_0 \re {w_0} v_1 \re{w_1} \dots) &=& \inf_{k\leq n_X} \sum_{i=0}^{k-1} w_i,
    \end{array}
\]
where $X$ is a set of vertices, and $n_X$ is the least $i$ such that $v_i \in X$ (if there is no such vertex, $n_X = \infty$, in which case the $\sup$ or $\inf$ ranges over all natural numbers $k$).
In words, when playing a mean-payoff game (i.e. a game with valuation $\MP$), the players seek to optimise the long term average.\footnote{Because of positional determinacy and finiteness of the game, the use of limsup rather than liminf is irrelevant in that it does not affect the value.}
When playing a game with valuation $\supsum^X$, the players seek to optimise the peek of the profile before $X$ is reached (at which point one may think that the game stops).
If $X$ is not reached then the game lasts forever, in particular we may have value $\infty$.
Similarly for $\infsum^X$.
Since the empty path, of weight $0$, is also taken into account, note that $\supsum^X$ takes only $\geq 0$ values and $\infsum^X$ takes $\leq 0$ values.
When $X$ is empty, we just write $\sup$ or $\inf$.
In the literature, such games (with empty $X$) are usually called energy games, and were introduced by~\cite{CAHS03} and further studied and popularised by~\cite{BFLMS08,BCDGR11}.
We now formally state positional determinacy.

\begin{theorem}[\cite{EM79,BFLMS08}]
    For any game, any subset $X$ of vertices, any $\val \in \{\MP, \supsum^X, \infsum^X\}$, and any vertex $v$, it holds that
    \[
        \inf_\sigma \sup_{\pi \models \sigma} \val(\pi) = \sup_\tau \inf_{\pi \models \tau} \val(\pi),
    \]
    where $\sigma$ and $\tau$ respectively range over strategies of Min and Max, and $\pi$ ranges over paths starting from $v$.
\end{theorem}
The value from the theorem is called the $\val$-value of the vertex $v$, or just its value if $\val$ is clear from context.
Moreover, optimal strategies (which clearly exist since there are finitely many strategies) can be chosen independently of the starting vertex.
Thanks to the theorem, we may think of these games as cyclic ones; for instance, in a mean-payoff game, Min and Max fix optimal strategies, and then the value of a vertex is just the average weight of the unique cycle that can be reached from $v$ when following the two strategies.

In the mean-payoff threshold problem, we ask, given a game, which vertices have value $\leq 0$ and which vertices have value $>0$.
These are called the winning regions of the game.
The more general value problem, which consists in determining the exact value of the vertices (which is a rational number in $[-W,W]$ whose denominator is $\leq n$), reduces to the threshold problem by dichotomy: to test whether the value is $\leq x$, it suffices to solve the threshold problem in the game where all weights are shifted by $-x$.

As was first formally observed by~\cite{BCDGR11}, an easy consequence of the theorem is that vertices with mean-payoff value $<0$ have finite $\sup$-value and minus infinite $\inf$-value, and likewise, vertices with mean-payoff value $>0$ have infinite $\sup$-value and finite $\inf$-value.
Vertices with mean-payoff value $=0$ have both $\sup$ and $\inf$-values finite, and they can be tedious to work with because they might break symmetry (for instance, in the above definition of the mean-payoff threshold problem).
Fortunately, we may get rid of them without loss of generality, simply by replacing each weight $w$ by $nw +1$: this preserves the sign of all the nonzero cycles (and therefore the sign of nonzero values), and turns all cycles of weight zero into positive ones (and therefore, vertices of value zero now get a positive value).
Throughout the paper, we will work under the assumption that there are no cycles of weight zero.\footnote{Note that reducing to such games replaces $W$ by $nW+1$ which does not affect combinatorial bounds (i.e.~bounds independent on $W$).}

To sum up, we have two kinds of vertices: those with mean-payoff value $<0$, finite $\sup$-value, and minus infinite $\inf$-value, and those with mean-payoff value $>0$, infinite $\sup$-value, and finite $\inf$-value.
Our goal is to determine which vertex is of which kind.
We refer to this partition as the winning regions in the game.
As explained in the introduction, almost all known algorithms (to the best of our knowledge, all of them except~\cite{ZP96}) actually compute (explicitly or implicitly) the $\sup$-values (sometimes called energy values), in order to determine whether or not they are finite.

Next is a useful technical result about $\supsum$. (Naturally a similar statement holds for $\infsum$.)

\begin{lemma}\label{lem:technical-sup}
    In a game, we have
    \[
        \supsum^X(v_0 \re{w_0} v_1 \re{w_1} v_2 \re{w_2} \dots) = w_0 + \supsum^X(v_1 \re{w_1} v_2 \re{w_2}\dots)
    \]
    as long as $v_0 \notin X$ and the right hand side is $\geq 0$.
\end{lemma}

\begin{proof}
    Since $v_0\notin X$ we have $n_X>0$ therefore
    \[
        \begin{array}{lcl}
            \supsum^X(v_0 \re{w_0} v_1 \re{w_1} \dots) &=& \sup_{k\leq n_X} \sum_{i=0}^{k-1} w_i \\ &=& \max(0, \sup_{1 \leq k\leq n_X} \sum_{i=0}^{k-1} w_i) \\ & = & \max(0,w_0 + \supsum(v_1 \re{w_0} v_2 \re{w_2} \dots)) \\&= & w_0+  \supsum(v_1 \re{w_0} v_2 \re{w_2} \dots) 
        \end{array}
    \]
\end{proof}

We now introduce some additional concepts and terminology which are not part of the usual background but crucial to the paper.
An edge $vv'$ is immediately optimal if $v$ is a Min vertex and $vv'$ has minimal weight among edges outgoing from $v$, or $v$ is a Max vertex and $vv'$ has maximal weight among edges outgoing from $v$.
Given a game $G$, let $N,Z,P$ denote the sets of vertices for which the weight of an immediately optimal edge is $<0,=0,>0$ respectively.
Let $ZN$ (resp. $ZP$) denote the set of vertices such that Min (resp. Max) can force to see a $<0$ (resp. $>0$) edge before a $>0$ (resp. $<0$) one.
(The name $ZN$ comes from ``zero then $N$''.)
Note that $N \subseteq ZN$ and $P \subseteq ZP$, and that $ZN$ and $ZP$ partition the vertices (since there is no zero cycle).
If the game is not clear from context, we add it as a subscript, e.g. $P_G$ or $ZN_G$.
We call $N,Z,P,NZ$ and $ZP$ the zones of the game; they can be computed in $O(m)$ by standard backtracking (in standard parity-game terminology, $ZN$ is the Min-attractor over $\leq 0$-edges to $<0$ edges).

We say that a vertex from $ZN$ is reduced if Min can force to immediately see a $\leq 0$ edge towards $ZN$, and likewise a vertex from $ZP$ is reduced if Max can force to immediately see a $\geq 0$ edge towards $ZP$.
Note that vertices in $Z$ are always reduced, but not necessarily those in $N$ and $P$.
We say that a game is reduced if all vertices are reduced.
Note that in a reduced game, from $ZN$ Min has a strategy ensuring to only visit $\leq 0$ edges while forever staying in $ZN$ (in particular it is a trap for Max).
Therefore, in a reduced game, vertices in $ZN$ have mean-payoff value $<0$, and similarly vertices in $ZP$ have value $>0$.
Observe that any game such that $ZN$ or $ZP$ is empty (or equivalently, $N$ or $P$ is empty), is reduced; we say that the game is positively reduced if $ZN$ is empty, and negatively reduced otherwise.
Given a subset of vertices $X$, we say that a game is reduced (resp. positively or negatively reduced) over $X$ if $G \cap X$ is reduced (resp. positively or negatively reduced).

A potential is a map $\phi: V \to \Z$.
The associated potential reduction turns the game $G$ into $G_\phi$ by replacing each weight $w(vv')$ by its modified weight $w_\phi(vv')$ defined by $w_\phi(vv') = w(vv') + w(v') - w(v)$.
Note that weights of cycles are not altered (in particular, (1) the mean-payoff values and winning regions stay the same, and (2) there are still no cycles of weight zero in $G_\phi$),
and note that the profile of a path goes to infinity (or minus infinity) if and only if this is the case also for modified weights.
We write $N_\phi$ as a shorthand for $N_{G_\phi}$ when $G$ is clear from context (and likewise for $P,Z,ZN,ZP$).
A potential $\phi$ is called a reducing potential if $G_\phi$ is reduced, it is called a reducing potential over a subset $X$ of vertices if $G_\phi$ is reduced over $X$, and likewise for positively or negatively reducing.

\begin{theorem}[\cite{GKK88, OhlmannGKK}]
    Every game admits a reducing potential $\phi$.
    Moreover, the obtained partition into $ZN_\phi$ and $ZP_\phi$ is independent of the reducing potential $\phi$.
\end{theorem}

By the discussion above, the obtained partition clearly coincides with the mean-payoff winning regions (and in particular, it is independent of $\phi$).
Existence of a reducing potential is also an easy proof: an example of a reducing potential is the potential defined by assigning, to vertices with mean-payoff value $<0$, their (positive, finite) $\sup$-value, and to vertices with mean-payoff value $>0$, their (negative, finite) $\inf$-value.
As alluded to earlier, our algorithm will not necessarily compute this potential.

\section{Presentation of the algorithm}\label{sec:algorithmpresentation}

We first try to explain some of the intuitions behind the algorithm; it could be that it's actually easier for the reader to skip the next part and jump to the more formal presentation of the algorithm, which is just roughly one page long (or even, maybe for some readers it's easier to go directly to the pseudocode in Section~\ref{sec:pseudocode}).

\paragraph*{Informal overview of the algorithm}

Let us turn our attention to the set of vertices $N$.
These are the vertices $v$ where Min can ensure that the first edge which is visited has weight $<0$.
In a best case scenario (for Min), she can even ensure that from $v$, the profile starts with a negative edge and then never again takes a $>0$ value; in this case the $\sup$-value of $v$ is $0$ (and thus its mean-payoff value is $<0$, but from now on we will no longer mention mean-payoff values).
For the purpose of this explanation, say that such a vertex in $N$ is strongly negative.

Now observe that conversely, as long as Min's winning region is non-empty, there exists a strongly negative vertex.
This is because from a vertex $v$ which is winning for Min, if both players play optimally, the vertex $v'$ where the peak of the profile is reached has to be strongly negative, since the profile will never again exceed it.

Now of course, not all vertices in $N$ are strongly negative (possibly even none of them).
But as a first approximation (which plays in the favor of Min), the algorithm will start by imagining that all vertices in $N$ are strongly negative.
Under this assumption, it is reasonable to expect that one may be able to compute the $\sup$ values by somehow backtracking from $N$; this corresponds in fact to computing the $\supsum^N$ values, which will be the first focus of the algorithm.
(Note here that we make an asymmetric choice of favoring $N$ over $P$, but ignore this for the moment.)
We thus maintain a set $F$ of vertices for which we have already computed the $\supsum^N$-values; initially we set $F$ to be $N$, where the $\supsum^N$ value is $0$.

$(*)$ The first bit of backtracking is quite straightforward: it is easy to compute the $\supsum^N$-value of vertices which have all their edges pointing towards $F$.
Doing this recursively, we may find the $\supsum^N$ value of all vertices whose paths all visit $F$, and add them to $F$.

The next bit is trickier.
Note that removing $F$ leaves a subgame, let us call it $H$.
We call the algorithm recursively on $H$, which breaks $H$ into its winning regions $H^-$ and $H^+$ and gives a reducing potential $\phi_H$ for $H$
(See Figure~\ref{fig:fig1}).
Now the goal will be to find a single vertex from $H$ for which we are able to determine the $\supsum^N$ value.

\begin{figure}[h]
    \begin{center}
        \includegraphics[width=0.55\linewidth]{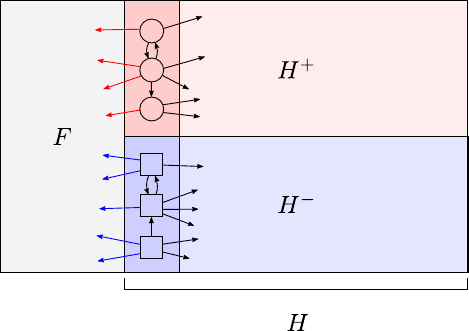}
        \caption{An illustration of the situation. Circles are Min-vertices and squares are Max-vertices. Red edges are escapes from $H^+$ and blue edges are escapes from $H^-$.}\label{fig:fig1}
    \end{center}
\end{figure}

To find such a vertex, we focus on those vertices which have means of escaping $H$, meaning Min-vertices in $H^+$ and Max-vertices in $H^-$ which have some edges towards $F$.
Therefore the goal is to find an escape which is optimal for $\supsum^N$.
The key insight here is that the reducing potential $\phi_H$ somehow levels out the vertices in $H$, allowing to compare the different possibilities of escaping, and being able to determine, for a single vertex $v$, that the player who controls this vertex is better-off escaping to $F$.

The first case is when $H^+$ is non-empty.
In the case where there is no edge from $H^+ \cap \Vmin$ to $F$ (the red edges in Figure~\ref{fig:fig1}), then vertices in $H^+$ have $\sup_N$-value $\infty$ in $G$ and therefore also $\sup$-value $\infty$, and moreover $\phi_H$ provides a reducing potential over these vertices.
Therefore, by a standard backtracking step, we compute the set $A \supseteq H^+$ of vertices of $G$ such that Max can ensure that the game reaches $H^+$ from $A$, together with a reducing potential over $A$, and then we may safely remove $A$ from the game (this is justified in Lemma~\ref{lem:gluingpotentials}), and recurse over $G \setminus A$.

Assume now that we have edges from $H^+ \cap \Vmin$ towards $F$.
In this case, clearly Min should escape $H^+$, since this gives finite $\sup_N$-value, whereas staying in $H^+$ grants her infinite $\sup_N$-value.
Using the reducing potential $\phi_H$ we may determine the best escaping edge $vv'$ (with $v \in H^+ \cap \Vmin$ and $v' \in F$), which turns out to be the one minimising $w(vv')+\supsum^N(v')-\phi_H(v)$ (see Lemma~\ref{lem:main1}a).
Therefore we may add $v$ to $F$ and continue our backtracking, meaning, go back to $(*)$.

In the case where $H^+$ is empty (and $H^-$ is not), we will see that there are necessarily some edges from $H^- \cap \Vmax$ to $F$ (the blue edges in Figure~\ref{fig:fig1}), and that the optimal escaping edge $vv'$ from $H^- \cap \Vmax$ to $F$ is the one maximising $w(vv') + \supsum^N(v')-\phi_H(v)$ (see Lemma~\ref{lem:main1}b).
Hence we proceed as previously, adding $v$ to $F$ and going back to $(*)$.

In this fashion, after a total of $\leq n$ recursive calls over smaller games, we either exit the loop at $(**)$, in which case we have reduced a part $A$ of the game which is winning for Max and may safely recurse on $G \setminus A$, or we have successfully computed the $\supsum^N$-value over $G$, and they are finite.

At this stage, since the $\supsum^N$-values provide a lower bound to the $\sup$-values, we may re-examine our assumption about strongly negative vertices, for instance if a Max vertex $v$ in $N$ has a $-1$ edge towards a vertex whose $\supsum^N$ value is $2$, then clearly $v$ cannot be strongly negative.
A nice way to do this is simply to consider the potential assigning its $\supsum^N$-value to each vertex: this has the effect of potentially turning some of the vertices in $N$ to $P$ or $Z$, while turning all vertices out of $N$ to $Z$.
(More formally, if $\phi$ denotes this potential, it is not hard to see that $N_\phi$ and $P_\phi$ are included in $N_G$ (Lemma~\ref{lem:main4}).)
Therefore either $|N|+|P|$ shrinks, or one of $N$ and $P$ becomes empty, in which case we are done because the game is then reduced.

Recall that we made an arbitrary choice of focussing on $N$ and computing $\supsum^N$, but we could alternatively focus on $P$ and compute $\infsum^P$ using the dual algorithm.
Naturally, this leads to the symmetrical outcome: either we find a part which is winning for Min (and a reducing potential for this part), or we find a potential $\phi$ (namely, $\phi=\infsum^P$) such that in $N_\phi$ and $P_\phi$ are included in $P$.

We will choose to compute $\supsum^N$ if $|N| \leq |P|$ and $\infsum^P$ otherwise.
This case disjunction is applied to every recursive call, which is why the algorithm is symmetric.

\paragraph*{Formal presentation of the algorithm}

We now repeat the above presentation in a less verbose style.
A more formal pseudo-code is given in Section~\ref{sec:pseudocode}.

\begin{enumerate}
    \item Compute the zones $N,Z,P,ZN$ and $ZP$, and check whether the game is reduced.
    If this is the case, we stop.
    \item Choose whichever of $N$ and $P$ is smallest; here we assume $|N| \leq |P|$ for concreteness (otherwise, dualise everything). 
    The goal of the subsequent steps will be to compute $\supsum^N$ over $G$, together with a reducing potential over vertices with infinite $\supsum^N$-value.
    Throughout the following steps we will update a set $F$ of vertices such that the value of $\supsum^N$ is already known for vertices in $F$.

    \item We may already add $N$ to $F$ since $\supsum^N$ is zero over $N$.

    \item\label{item:back} Now we deal with the set $S$ of vertices not in $F$ from which all paths lead to $F$: it is easy to compute their $\supsum^N$ value by backtracking, and then add them to $F$.
    
    \item\label{item:loop} Let $H$ be the subgame obtained by removing from $G$ all vertices in $F$ (note that this is indeed a subgame thanks to the previous step).
    Call the algorithm recursively on $H$ thereby obtaining a reducing potential $\phi_H$ for $H$ with corresponding winning regions $H^-$ and $H^+$ for Min and Max, respectively.
    There are a few cases.
    \begin{enumerate}
        \item\label{item:ZPHnonempty} If $H^+$ is non-empty.
        \begin{enumerate}
            \item\label{item:escapefromZPH} If there is an edge from $H^+ \cap \Vmin$ to $F$.
            Then let $vv'$ be such an edge minimising $w(vv')+\supsum^N(v')-\phi_H(v)$.
            We claim that $\sup_N(v)=\sup_N(v')+w(vv')$ (see Lemma~\ref{lem:main1}a), therefore we may add $v$ to $F$, and then go back to step~\ref{item:back}.
            \item\label{item:noescapefromZPH} If there is no edge from $H^+ \cap \Vmin$ to $F$.
            Then Max may play optimally in $H^+$ while forcing the game to remain inside of it, thus the $\sup$-values over $H^+$ are $\infty$: we have found a part of the winning region for Max in $G$.
            By a standard backtracking, we compute the set $A \supseteq H^+$ of vertices from which, in $G$, Max can ensure that the game reaches $H^+$.
            Note that $\phi_H$ provides a reducing potential over $H^+$, and it is easily extended to a reducing potential over $A$ during the backtracking step.
            We may therefore recurse on $G \setminus A$ (see Lemma~\ref{lem:gluingpotentials}), which is easily seen to indeed be a subgame.
        \end{enumerate}
        \item\label{item:ZPHempty} Otherwise, $H^+$ is empty.
        \begin{enumerate}
            \item\label{item:escapefromZNH}
            If $H^-$ is non-empty, then there is an edge from $H^- \cap \Vmax$ towards $F$ (see first part of Lemma~\ref{lem:main1}b).
            We let $vv'$ be such an edge maximising $w(vv')+\supsum^N(v')-\phi_H(v)$.
            We claim that $\supsum^N(v)=\supsum^N(v') + w(vv')$ (see second part of Lemma~\ref{lem:main1}b), therefore we may add $v$ to $F$ and then go back to step~\ref{item:back}.
            \item\label{item:noescapefromZNH} Otherwise, $H$ is empty.
            Therefore we have finished computing $\supsum^N$-values over $G$, and they are all finite.
            We let $G'$ be obtained from performing the corresponding potential reduction over $G$, observe that $N_{G'}$ and $P_{G'}$ are both included in $N_G$ (see Lemma~\ref{lem:main4}), and recurse on $G'$.
        \end{enumerate}
    \end{enumerate}
    Since the size of $F$ increases every time~\ref{item:escapefromZPH} or~\ref{item:escapefromZNH} occurs, we must reach either~\ref{item:noescapefromZPH} or~\ref{item:noescapefromZNH} after a linear number of recursive calls performed on strictly smaller games.
    If~\ref{item:noescapefromZPH} is reached then the next recursive call is on a strictly smaller game.
    If~\ref{item:noescapefromZNH} is reached then (see Lemma~\ref{lem:main4}), the next recursive call is performed on a game $G'$ such that either $N_{G'}$ or $P_{G'}$ is empty, in which case $G'$ is reduced so the algorithm stops, or the sizes of $N_{G'}$ and $P_{G'}$ are both smaller than the smallest among the sizes of $N_G$ and $P_G$, which guarantees termination.
\end{enumerate}

We also give a more formal pseudocode in Section~\ref{sec:pseudocode}.

\section{Proofs for the main lemmas}\label{sec:proofs}

We go on to proving the three main lemmas that are used for the algorithm.
The first lemma explains how to glue together reducing potentials under some assumptions, justifying the recursive call in step~\ref{item:noescapefromZPH}.

\begin{lemma}\label{lem:gluingpotentials}
    Let $G$ be a game and $A$ be a trap for Min, let $\phi_A$ be a positively reducing potential over $A$, let $G' = G \setminus A$ which is assumed to be a trap for Max, and let $\phi'$ be a reducing potential over $G'$.
    Let $\delta$ be a constant such that for every $v' \in G'$ and $a \in A$ we have
    \[
        \delta \geq -w(v'a) - \phi_A(a) + \phi'(v'),
    \]
    and let $\phi=(\phi_A + \delta) \uplus \phi'$ be the potential which coincides with $\phi_A + \delta$ over $A$ and with $\phi'$ over $G'$.
    Then $\phi$ is a reducing potential for $G$.
\end{lemma}

\begin{proof}
    First, note that $G'_{\phi_A}$ and $G'_{\phi_A + \delta}$ are the same games, therefore $G'_{\phi_A + \delta}$ is reduced, and the winning regions in both games coincide.
    We claim that $ZN_\phi=ZN_{\phi'}$ and $ZP_\phi=ZP_{\phi'} \cup A$.
    We first observe that in $G_\phi$, Max has a strategy, from $ZP_{\phi'} \cup A$, to ensure that the first edge which is seen is $\geq 0$ and towards $ZP_{\phi'} \cup A$.

    Indeed, from $A$ Max may even force a $\geq 0$ edge towards $A$, this is because $\phi_A$ (and likewise $\phi_A + \delta$) is positively reducing, $A$ is a trap for Min, and $\phi$ coincides with $\phi_A + \delta$ over $A$.
    Now from $ZP_{\phi'}$, Max has a strategy to either see an edge of $\phi$-modified weight $\geq 0$ towards $ZP_{\phi'}$ (since $\phi'$ is reducing in $G'$ and coincides with $\phi$), or to see an edge $v'a$ towards $A$.
    In the latter case, the edge $v'a$ has $\phi$-modified weight $w(v'a) + \phi_A(a) + \delta - \phi'(v')$, which is $\geq 0$ thanks to our assumption on $\delta$.

    Likewise, we should also prove that from $ZN_{\phi'}$, Min has a strategy ensuring to first see an edge of $\phi$-weight $\leq 0$ towards $ZN_{\phi'}$.
    This is clear since $G'$ is a trap for Max, $\phi'$ is reducing over $G'$, and $\phi$ and $\phi'$ coincide on $G'$.
\end{proof}

The next crucial lemma explains how to obtain an optimal escaping edges when a reducing potential is provided.

\begin{lemma}\label{lem:main1}
    Let $G$ be a game, let $F$ be a set of vertices containing $N$ and such that $\supsum^N$ is finite over $F$, let $H$ be obtained by removing $F$ from $G$ and assume that $H$ is a subgame.
    Let $\phi_H$ be a reducing potential for $H$ and let $H^+$ and $H^-$ denote the winning regions.
    \begin{enumerate}[a.]
        \item Assume that there is an edge from $H^+ \cap \Vmin$ to $F$ and let $vv'$ be such an edge minimising $w(vv') + \supsum^N(v') - \phi_H(v)$.
        Then $\supsum^N(v)=\supsum^N(v') + w(vv')$.
        \item Assume that $H^+$ is empty, and that $H^-$ is not.
        Then there is an edge $vv'$ from $v \in H^- \cap \Vmax$ to $v' \in F$.
        Let $vv'$ be such an edge maximising $w(vv') + \supsum^N(v') - \phi_H(v)$.
        Then $\supsum^N(v)=\supsum^N(v') + w(vv')$.
    \end{enumerate}
    
\end{lemma}
\begin{proof}[Proof of Lemma~\ref{lem:main1}a]
    Since Min can play the edge $vv'$ and then play optimally from $v'$, we have 
    \[
        \supsum^N(v) \leq w(vv') + \supsum^N(v'),
    \]
    because $w(vv') \geq 0$ (this is because $v \notin N$ since $N \subseteq F$) and therefore the right hand side is $\geq 0$, so Lemma~\ref{lem:technical-sup} applies.
    So we focus on the other inequality.
    We let Max play the strategy, over $H^+$, that uses only edges towards $H^+$ with $\phi_H$-modified weight $\geq 0$, and over $F$, he plays optimally. 
    Starting from $v$, this forces the play to remain within $F \cup H^+$, so we do not need to define Max's strategy over $H^-$.
    We claim that this strategy achieves the wanted value from $v$.

    Indeed, if the play remains in $H^+$, then only edges with $\phi_H$-modified weight $\geq 0$ are visited, and since there are no zero cycles, the profile of modified weights goes to $\infty$, and hence this is also true with the original weights.
    Otherwise, the play remains in $H^+$ until visiting an edge $uu'$ from $H^+ \cap \Vmin$ to $F$.
    Now the prefix of the play from $v$ to $u$ has $\phi_H$-modified weight $\geq 0$, so we get
    \[
        \begin{array}{lcl}
        \text{value of play from $v$} &\geq& \text{value of prefix } + w(uu') + \supsum^N(u') \\
        &=& \text{modified weight of prefix} + \phi_H(v) - \phi_H(u) + w(uu') + \supsum^N(u') \\
        &\geq & \phi_H(v) + w(uu') + \supsum^N(u') - \phi_H(u) \\
        & \geq & \phi_H(v) + w(vv') + \supsum^N(v') - \phi_H(v) = w(vv') + \supsum^N(v'),
        \end{array}
    \]
    as required.
\end{proof}

Next is a technical statement that will be used to prove Lemma~\ref{lem:main1}b.

\begin{lemma}\label{lem:technical}
    Let $H$ be a game such that the winning region of Max is empty, and let $\phi$ be a reducing potential.
    There exists a vertex in $N_H$ such that $\phi(v)$ is minimal.
\end{lemma}

\begin{proof}
    Since $\phi$ is reducing and the winning region of Max is empty, we have $P_\phi=\varnothing$, i.e.~every Min vertex has an outgoing edge with modified weight $\leq 0$ and every edge outgoing from a Max vertex has modified weight $\leq 0$.
    By contradiction, assume vertices $v$ in $N_H$ do not have minimal $\phi(v)$ value, and let $u$ be a vertex such that $\phi(u)$ is minimal.
    We proceed as follows.
    \begin{itemize}
        \item If $u\in \Vmin$.
        Then let $uu'$ be an edge of $\phi$-modified weight $\leq 0$.
        We have 
        \[
            w(uu') = \underbrace{w_\phi(uu')}_{\leq 0} + \underbrace{\phi(u) - \phi(u')}_{\leq 0} \leq 0. \tag{*}\label{eq}
        \]
        If $w(uu')<0$ then $u \in N_G$ which is a contradiction.
        Otherwise, $w(uu')=0$ so $\phi(u')=\phi(u)$ and we continue the process from $u'$.
        \item If $u \in \Vmax$.
        Then all edges $uu'$ have $\phi$-modified weight $\leq 0$ and satisfy~(\ref{eq}).
        If they all satisfy $w(uu')<0$ we have $u \in N_G$ and we are done.
        Otherwise, there is an edge $uu'$ with $w(uu')=0$ thus $\phi(u')=\phi(u)$ and we continue from $u'$.
    \end{itemize}
    The above process must stop since there are no zero cycles, which proves the lemma.
\end{proof}

We are now ready for the next proof.

\begin{proof}[Proof of Lemma~\ref{lem:main1}b]
    We first prove the claim that if $H^+$ is empty and $H^-$ is not, then there is an edge from $H^+ \cap \Vmax$ to $F$.
    In fact, we prove a stronger claim that will be useful later on: there is a vertex $v \in \Vmax \cap N_H$ (i.e. all edges from $v$ to $H=H^+$ are $<0$) with a $\geq 0$ edge towards $F$.
    Indeed, $H^-$ is non-empty therefore $N_H$ is non-empty and since $N_G \subseteq F$, all Min vertices in $H$ have only $\geq 0$ outgoing edges and thus cannot belong to $N_H$. 
    Hence $N_H$ is exactly comprised of Max vertices whose edges towards $H$ are all $<0$.
    Moreover, a vertex in $N_H$ must have a $\geq 0$ outgoing edge otherwise it would belong to $N_G$, and such an edge can only be towards $F$.

    Now we prove the claim that if $vv'$ maximises $w(vv') + \supsum^N(v') - \phi_H(v)$ from $H \cap \Vmax$ to $F$, then $\supsum^N(v)=\supsum^N(v') + w(vv')$.
    Since Max can play the edge $vv'$ and then play optimally from $v'$, we have $\supsum^N(v) \geq \supsum^N(v')+w(vv')$: this is obvious if the right hand side is $\leq 0$, and otherwise follows from Lemma~\ref{lem:technical-sup}.
    So we focus on the other inequality.
    
    We now apply Lemma~\ref{lem:technical} to $H$: there is a vertex $u$ in $N_H \subseteq \Vmax$ for which $\phi_H(u)$ is minimal (among all vertices in $H$).
    Recall that by the maximality assumption of the lemma, it holds for all edges $uu'$ with $u' \in F$ that
    \[
        \begin{array}{l}
        w(vv') + \supsum^N(v') - \phi_H(v) \geq w(uu') + \supsum^N(u') - \phi_H(u),
        \end{array}
    \]
    and there is an edge $uu'$ such that $w(uu') \geq 0$ (and also it always holds that $\supsum^N(u')\geq 0$), therefore we have
    \[
        \begin{array}{l}
        w(vv') + \supsum^N(v') - \phi_H(v) \geq - \phi_H(u).
    \end{array} \tag{*}\label{eq2}
    \]

    We let Min play the strategy, over $H^-$, that uses only edges towards $H^-$ with $\phi_H$-modified weight $\leq 0$, and over $F$, she plays optimally.
    We claim that this strategy achieves the wanted value from $v$.

    Indeed, while the play remains in $H^-$, only edges with $\phi_H$-modified weight $\leq 0$ are visited, and therefore for a prefix of the play ending in some vertex $v' \in H^-$ we get
    \[
        \begin{array}{lcl}
        \text{weight of the prefix} &=& \underbrace{\text{modified weight of the prefix}}_{\leq 0} + \phi_H(v) \underbrace{- \phi_H(v')}_{\leq -\phi_H(u)} \\
        &\leq & \phi_H(v) - \phi_H(u) \underbrace{\leq}_{\text{by~(\ref{eq2})}} w(vv') + \supsum^N(v'),
        \end{array}
    \]
    thus the value of the play is $\leq w(vv') + \supsum^N(v')$ as required.
    (We free the symbol $u$ for the remainder of the proof.)
    
    If however the play exits $H^-$ towards $F$, then it must be via an edge $uu'$ with $u$ in $\Vmax$, so the play is of the form:
    \[
        \underbrace{v=h_0 \re{} h_1 \re{} \dots \re{} h_n =u}_{\text{inside } H^-} \re{} \underbrace{u'}_{\text{in } F} \re{} \dots,
    \]
    where the suffix of the play starting from $u'$ follows the Min strategy which is optimal from $u'$ therefore has value $\leq \supsum^N(u')$.
    Therefore the value of the play is less or equal to
    \[
        \max\big[\max_i (\text{weight of prefix up to $h_i$}), (\text{weight of prefix up to $u'$}) + \supsum^N(u')\big],
    \]
    and we have already argued above that the weight of the prefix up to $h_i$ is $\leq w(vv') + \supsum^N(v')$.
    To upper bound the other term, we proceed similarly as in Lemma~\ref{lem:main1}a: the prefix of the play up to $u$ has $\phi_H$-modified weight $\leq 0$, so we get
    \[
        \begin{array}{lcl}
       && \text{weight of prefix up to $u'$} + \supsum^N(u') \\&\leq& \text{weight of prefix up to $u$} + w(uu') + \supsum^N(u') \\
        &=& \text{modified weight of prefix up to $u$} + \phi_H(v) - \phi_H(u) + w(uu') + \supsum^N(u') \\
        &\leq & \phi_H(v) + w(uu') + \supsum^N(u') - \phi_H(u) \\
        & \leq & \phi_H(v) + w(vv') + \supsum^N(v') - \phi_H(v) = w(vv') + \supsum^N(v'),
        \end{array}
    \]
    concluding the proof.
\end{proof}



This last easy lemma justifies the termination of the algorithm.

\begin{lemma}\label{lem:main4}
    Let $G$ be a game over which $\supsum^N$ is finite, and let $G'$ be obtained by applying the corresponding potential reduction.
    Then both $N_{G'}$ and $P_{G'}$ are contained in $N$.
\end{lemma}
\begin{proof}
    Let $v$ be a Min vertex in $G'$ which is not in $N$.
    Then for every outgoing edge $vv'$ we have $\supsum^N(v) \geq \supsum^N(v') + w(vv')$ and this is an equality for at least one outgoing edge.
    Therefore, the $\supsum^N$-modified weights are all $\geq 0$ and at least one of them is $=0$, so $v$ belongs to $Z_{G'}$.
    The proof is analogous for Max vertices.
\end{proof}

\section{Pseudocode}\label{sec:pseudocode}

\SetKwFunction{FComputeZones}{ComputeZones}
\SetKwFunction{FReduceGame}{ReduceGame}
\SetKwFunction{FPotentialReduction}{PotentialReduction}
\SetKwFunction{FAttractMaxAndReduce}{AttractMaxAndReduce}
\SetKwFunction{FAttractMinAndReduce}{AttractMinAndReduce}
\SetKwFunction{FBacktrackAllPaths}{BacktrackAllPaths}
\SetKwFunction{FBacktrackInitialisationMin}{BacktrackInitialisationMin}
\SetKwFunction{FBacktrackInitialisationMax}{BacktrackInitialisationMax}
\SetKwFunction{FReduced}{Reduced}
\SetKwFunction{FTrue}{True}
\SetKwFunction{FFalse}{False}
\SetKw{KwBreak}{break}
\SetKw{KwContinue}{continue}
\SetKw{Repeat}{do}
\SetKw{Until}{until}
\SetKwProg{Fn}{Function}{}{}

Here is some pseudocode for our algorithm. All games are assumed without zero cycles.
Given maps $f:A \to X$ and $g:B \to X$, where $A$ and $B$ are disjoint sets, we let $f \uplus g$ denote the map from $A \cup B$ to $X$ which coincides with $f$ on $A$ and with $g$ on $B$.

\begin{algorithm}[H]
    \DontPrintSemicolon

    \Fn{\FReduceGame{$G$}}{
        \KwIn{A game $G=(V,E,w)$}
        \KwOut{A tuple $(G^-,G^+,\phi)$ where $\phi$ is a reducing potential and $G^-,G^+$ are the corresponding winning regions}
        $(N,P,ZN,ZP) \gets \FComputeZones(G)$\;
        \If{$\FReduced(G,N,P,ZN,ZP)$}{
            \Return{(ZN,ZP,0)}
        }
        \If{$|N| \leq |P|$\label{line:choice}}{
            $F \gets N$\; \label{line:initialisation-of-F}
            forall $v$, $\supsum^N(v) \gets 0$\;
            \While{$\FTrue$}{
                $\FBacktrackAllPaths(G,F,\supsum^N)$ \; \label{line:backtrack-all-paths}
                $H \gets G \setminus F$\;
                $(H^-,H^+,\phi_H) \gets \FReduceGame(H)$\; \label{line:main-recurse}
                \If{$H^+ \neq \varnothing$}{
                    \If{there is an edge from $H^+ \cap \Vmin$ to $F$}{
                        let $vv'$ be such an edge minimising $w(vv') + \supsum^N(v') - \phi_H(v)$\; \label{line:start-H+}
                        $\supsum^N(v) \gets w(vv') + \supsum^N(v')$\;
                        $F \gets F \cup \{v\}$\; \label{line:add-v-H+}
                        \KwContinue\;
                    }
                    \Else{
                        $(A,\phi_A) \gets \FAttractMaxAndReduce(G,H^+,\phi_H)$\;
                        $G' = G \setminus A$\;
                        $(G'^+,G'^-,\phi') \gets \FReduceGame(G')$\;
                        $\delta = - \min_{v \in G', v' \in A} w(vv') - \min_{v' \in A} \phi_A(v') + \max_{v \in G'} \phi'(v)$\;
                        \Return{$(G'^-, G'^+ \cup A, (\phi_A + \delta) \uplus \phi')$}\;
                    }
                }
                \Else{
                    \If{$H^- \neq \varnothing$}{
                        let $vv'$ be an edge maximising $w(vv') + \supsum^N(v') - \phi_H(v)$ from vertices in $H \cap \Vmax$ to $F$\; \label{line:start-H-}
                        $\supsum^N(v) \gets w(vv') + \supsum^N(v')$\;
                        $F \gets F \cup \{v\}$\; \label{line:add-v-H-}
                        \KwContinue\;
                    }
                    \Else{
                        $G' \gets \FPotentialReduction(G,\supsum^N)$\;\label{line:potential-reduction}
                        \Return{$\FReduceGame(G')$} \;
                    }
                }
            }
        }
        \Else{\tcc{Proceed symmetrically}} 
    }\caption{Main recursive procedure}\label{alg:main}
    \end{algorithm}

\newpage

Here is a pseudocode for the three (similar) backtracking subprocedures, determining if a game is reduced, and computing a potential reduction.
Each of them runs in time $O(m)$.
We write $v^{-1} E = \{v' \in V \mid vv' \in E\}$ and $Ev^{-1}=\{v' \in V \mid v'v \in E\}$.

\begin{algorithm}[H]
    \DontPrintSemicolon
    \Fn{\FComputeZones{$G$}}{
        \KwIn{A game $G=(V,E,w)$}
        \KwOut{A tuple $(N,P,ZN,ZP)$ corresponding to the respective zones in the game.}
        $N \gets \{v \in \Vmin \mid \exists v' \in v^{-1} E, w(vv') <0\} \cup \{v \in \Vmax \mid \forall v'\in v^{-1} E, w(vv')<0\}$\;
        $P \gets \{v \in \Vmin \mid \forall v'\in v^{-1} E, w(vv') >0\} \cup \{v \in \Vmax \mid \exists v'\in v^{-1} E, w(vv')>0\}$\;
        $Z \gets V \setminus (N \cup P)$\;
        \ForAll{$v \in Z \cap \Vmax$}{
            $\esc[v] \gets |\{v' \in v^{-1} E \mid w(vv') = 0\}|$\;    
        }
        $ZN \gets \varnothing$\;
        $Q \gets N$\;
        \While{$Q \neq \varnothing$}{
            take some $v \in Q$\;
            $Q \gets Q \setminus \{v\}$\;
            $ZN \gets ZN \cup \{v\}$\;
            \ForAll{$v' \in Ev^{-1} \setminus (P \cup ZN)$}{
                \If{$v' \in \Vmin$ and $w(v'v) = 0$}{$Q \gets Q \cup \{v'\}$}
                \Else{
                    $\esc(v') \gets \esc(v')-1$\;
                    \If{$\esc(v')=0$}{$Q \gets Q \cup \{v'\}$}
                }
            }
        }
        \Return{$(N,P,ZN,V\setminus ZN)$}
    }
\end{algorithm}

\begin{algorithm}[H]
    \DontPrintSemicolon
    \Fn{\FReduced{$G,N,P,ZN,ZP$}}{
        \KwIn{A game $G=(V,E,w)$ together with its zones $N,P,ZN$ and $ZP$}
        \KwOut{$\FTrue$ if $G$ is reduced and $\FFalse$ otherwise.}
        $b \gets \FTrue$ \;
        $b \gets b  \text{ and } [\forall v \in \Vmin \cap N, \exists v' \in v^{-1}E \cap ZN, w(vv') \leq 0$] \;
        $b \gets b  \text{ and } [\forall v \in \Vmax \cap N, \forall v' \in v^{-1}E, v' \in ZN $] \;
        $b \gets b  \text{ and } [\forall v \in \Vmax \cap P, \exists v' \in v^{-1}E \cap ZP, w(vv') \geq 0$] \;
        $b \gets b  \text{ and } [\forall v \in \Vmin \cap P, \forall v' \in v^{-1}E, v' \in ZP$] \;
        \Return{$b$}
    }
\end{algorithm}

\begin{algorithm}[H]
    \DontPrintSemicolon
    \Fn{\FBacktrackAllPaths{$G,F,\phi$}}{
        \KwIn{A game $G=(V,E,w)$, a subset $F$ of vertices such that $N_G \subseteq F$, and a potential $\phi$ which coincides with $\supsum^{N_G}$ over $F$.}
        \KwOut{Updates $F$ and $\phi$ in place to incorporate vertices whose paths all lead to $F$ to $F$ so that $\phi$ still coincides with $\supsum^{N_G}$ over $F$.}
        \ForAll{$v \in V \setminus F$}{
            $\esc[v] \gets |\{v' \in v^{-1} E \mid v' \notin F\}|$\;
        }
        $Q \gets \{v \in V \setminus F \mid \esc(v)=0\}$\;
        \While{$Q \neq \varnothing$}{
            take some $v \in Q$\;
            $Q \gets Q \setminus \{v\}$\;
            $F \gets F \cup \{v\}$\;
            \If{$v \in \Vmin$}{
                $\phi(v)=\min_{v' \in v^{-1}E} w(vv') + \phi(v')$\;
            }
            \Else{
                $\phi(v)=\max_{v' \in v^{-1}E} w(vv') + \phi(v')$\;
            }
            \ForAll{$v' \in Ev^{-1} \setminus F$}{
                $\esc(v') \gets \esc(v')-1$\;
                \If{$\esc(v')=0$}{$Q \gets Q \cup \{v'\}$}
            }
        }
        \Return{}
    }
\end{algorithm}

\begin{algorithm}[H]
    \DontPrintSemicolon
    \Fn{\FAttractMaxAndReduce{$G,T,\phi_T$}}{
        \KwIn{A game $G=(V,E,w)$, a target set of vertices $T$ and a potential $\phi_T$ which is positively reducing over $T$.}
        \KwOut{A tuple $(A,\phi_A)$ where $A \supseteq T$ is the set of vertices such that Max can force the game to reach $T$, and $\phi_A$ is positively reducing over $A$. }
        $A \gets T$\;
        $\phi_A \gets \phi_T$\;
        \ForAll{$v \in (V \setminus T) \cap \Vmin$}{
            $\esc[v] \gets |\{v' \in v^{-1} E \mid v' \notin T\}|$\;
        }
        $Q \gets \{v \in (V \setminus T)  \cap \Vmin \mid \esc(v)=0\} \cup \{v \in (V \setminus T) \cap \Vmax \mid \exists z' \in v^{-1}E \cap T\}$\;
        \While{$Q \neq \varnothing$}{
            take some $v \in Q$\;
            $Q \gets Q \setminus \{v\}$\;
            $A \gets A \cup \{v\}$\;
            \If{$v \in \Vmin$}{
                $\phi(v)=\min_{v' \in v^{-1}E} w(vv') + \phi_A(v')$\;
            }
            \Else{
                take some $v' \in v^{-1}E \cap A$\;
                let $\phi_A(v)= w(vv') + \phi_A(v')$\;
            }
            \ForAll{$v' \in Ev^{-1}$}{
                \If{$v' \in \Vmin$}{
                    $\esc(v') \gets \esc(v')-1$\;
                    \If{$\esc(v')=0$}{$Q \gets Q \cup \{v'\}$}
                }
                \Else{
                    $Q \gets Q \cup \{v'\}$\;
                }
            }
        }
        \Return{$(A,\phi_A)$}
    }
 \end{algorithm}

\begin{algorithm}[H]
    \DontPrintSemicolon
    \Fn{\FAttractMinAndReduce{$G,T,\phi_T$}}{
        \tcc{Symmetric}
    }
\end{algorithm}

\begin{algorithm}[H]
    \Fn{\FPotentialReduction($G,\phi$)}{
        \KwIn{A game $G=(V,E,w)$ and a potential $\phi$.}
        \KwOut{The game $G'=(V,E,w_\phi)$ obtained by performing the corresponding potential reduction.}
        \ForAll{$vv' \in E$}{
            $w_\phi(vv') \gets w(vv') + \phi(v') = \phi(v)$\;
        }
        \Return{$(V,E,w_\phi)$}
    }
\end{algorithm}

\section{Optimisations and variants}\label{sec:optimisations-and-variant}

We propose some optimisations and variants for the algorithm.

\subsection{Optimisation: better initialisation of $F$}\label{sec:optim1}

Given a game $G$, consider the set $S$ of vertices from which Min can ensure that $N$ will be visited before any edge with weight $>0$ is seen.
Stated differently, $S$ is the largest set of vertices such that from $S$, Min can ensure that either only edges with weight $\leq 0$ are seen forever, or such edges are seen until $N$ is reached. (Note that $N \subseteq S$.)
Clearly, $S$ can be computed in polynomial time (by a backtracking procedure similar to the ones above), and moreover, $\supsum^N$ takes value $0$ on $S$.
The first optimisation replaces line~\ref{line:initialisation-of-F} in Algorithm~\ref{alg:main} by $F \gets \FBacktrackInitialisationMin(G,N,P)$ where this subprocedure computing $S$ can be implemented as follows.

\begin{algorithm}[H]
    \DontPrintSemicolon
    \Fn{\FBacktrackInitialisationMin{$G,N,P$}}{
        \KwIn{A game $G=(V,E,w)$ and its zones $N,P$.}
        \KwOut{The set of vertices $S$ from which Min can ensure that $N$ is seen before any $>0$ edge.}
        $A \gets \varnothing$\;
        \ForAll{$v \in \Vmin \setminus (P \cup N)$}{
            $\esc[v] \gets |\{v' \in v^{-1} E \mid v' \notin P \text{ and } w(vv') = 0\}|$\;
        }
        $Q \gets P$\;
        \While{$Q \neq \varnothing$}{
            take some $v \in Q$\;
            $Q \gets Q \setminus \{v\}$\;
            $A \gets A \cup \{v\}$\;
            \ForAll{$v' \in Ev^{-1} \setminus N$}{
                \If{$v' \in \Vmin$ and $w(vv') = 0$}{
                    $\esc(v') \gets \esc(v')-1$\;
                    \If{$\esc(v')=0$}{$Q \gets Q \cup \{v'\}$}
                }
                \Else{
                    $Q \gets Q \cup \{v'\}$\;
                }
            }
        }
        \Return{$V \setminus A$}
    }
 \end{algorithm}

 Naturally, a symmetrical procedure $\FBacktrackInitialisationMax$ is employed in the dual case.

\subsection{Optimisation: fixing many vertices at once}\label{sec:optim2}

The next optimisation is more involved.
The idea is to enhance lines~\ref{line:start-H+}--\ref{line:add-v-H+} and lines~\ref{line:start-H-}--\ref{line:add-v-H-} by an additional backtracking procedure in order to potentially add several vertices to $F$ at once, instead of only one.
This additional polynomial-time computation potentially removes many (expensive) recursive calls.

More precisely, instead of lines~\ref{line:start-H+}--\ref{line:add-v-H+}, we proceed as follows.
Assume that there is an edge from $H^+ \cap \Vmin$ to $F$ and let $m$ denote the minimal value of $w(vv') + \supsum^N(v') - \phi_H(v)$ for $vv'$ such an edge.
We say that an edge $vv'$ is a good escape if it is from $H^+$ to $F$ and $w(vv') + \supsum^N(v') - \phi_H(v) \leq m$ (note that an edge from $H^+ \cap \Vmax$ to $F$ could also be a good escape).
Then we compute the set $S$ of vertices in $H^+$ from which in $G$, Min can ensure that a good escape is seen before any edge which has either $\phi_H$-modified weight $>0$ or leaves $H^+$. 
Stated differently, $S$ is the largest set of vertices $\subseteq H^+$ such that from $S$, Min can ensure that the play only visits edges with $\phi_H$-modified weight $\leq 0$ while it remains in $H^+$, and either stays forever in $H^+$ or exits with a good escape.
Then for every vertex $v$ in $S$, the lemma below claims that $\supsum^N(v)=m+\phi_H(v)$, so we may fix $\supsum^N(v)$ and add all of $S$ to $F$.
For lines~\ref{line:start-H-}--\ref{line:add-v-H-} we proceed similarly.

We therefore establish an extension of Lemma~\ref{lem:main1} which is as follows.

\begin{lemma}
    Let $G$ be a game, let $F$ be a set of vertices containing $N$ and such that $\supsum^N$ is finite over $F$, let $H$ be obtained by removing $F$ from $G$ and assume that $H$ is a subgame.
    Let $\phi_H$ be a reducing potential for $H$ and let $H^+$ and $H^-$ denote the winning regions.
    \begin{enumerate}[a.]
        \item\label{item:H+} Assume that there is an edge from $H^+ \cap \Vmin$ to $F$ and let $m$ denote the minimal value of $w(vv') + \supsum^N(v') - \phi_H(v)$ for $vv'$ such an edge.
        We say that an edge $vv'$ is a good escape if it is from $H^+$ to $F$ and $w(vv') + \supsum^N(v') - \phi_H(v) \leq m$.
        Let $S$ be the set of vertices in $H^+$ from which in $G$, Min can ensure that a good escape is seen before any edge which has either $\phi_H$-modified weight $>0$ or leaves $H^+$.
        Then for every vertex $v$ in $S$ we have $\supsum^N(v)=m+\phi_H(v)$.

        \item\label{item:H-} Assume that $H^+$ is empty, and that $H^-$ is not.
        Then there is an edge $vv'$ from $v \in H^- \cap \Vmax$ to $v' \in F$ and let $m$ denote the maximal value of $w(vv') + \supsum^N(v')-\phi_H(v)$ for $vv'$ such an edge.
        Say that an edge $vv'$ is a good escape if it is from $H^-$ to $F$ and $w(vv')+ \supsum^N(v') - \phi_H(v) \geq m$.
        Let $S$ be the set of vertices in $H^-$ from which in $G$, Max can ensure that a good escape is seen before any edge which has either $\phi_H$-modified weight $<0$ or leaves $H^-$.
        Then for every vertex $v$ in $S$ we have $\supsum^N(v)=m+\phi_H(v)$.
    \end{enumerate}
\end{lemma}

\begin{proof}
    We start with the first item.
    We first prove that for every vertex $v \in S$ we have $\supsum^N(v) \leq m+\phi_H(v)$.
    Consider a Min strategy $\sigma$ which ensures that over $S$, a good escape is seen before any edge with $\phi_H$-modified weight $>0$ or any edge leaving from $H^+$ and from $F$, it is optimal for $\supsum^N$.
    We let $\supsum^N_\sigma(v)$ denote the value of $v$ assuming $\sigma$, i.e.~it is the supremum value of a play from $v$ consistent with $\sigma$; we say that a Max strategy $\tau$ is optimal against $\sigma$ if for every vertex $v$, the value of the unique play from $v$ consistent with both $\sigma$ and $\tau$ coincides with $\supsum^N_\sigma(v)$.

    We now prove the following claim: there is a Max strategy $\tau$ which is optimal against $\sigma$ and such that for every vertex $v \in S$, we have $\supsum^N_\sigma(v)=w(vv') + \supsum^N_\sigma(v')$, where $vv'=\tau(v)$ if $v \in \Vmax$ and $vv'=\sigma(v)$ if $v \in \Vmin$.
    Note that for $v \in \Vmin \cap S$, the wanted equality holds (regardless of $\tau$), since $w(vv') \geq 0$ (because $v \notin N$ since $N \subseteq F$) and thanks to Lemma~\ref{lem:technical-sup}.
    Consider a vertex $v \in \Vmax \cap S$.
    Since $v \notin N$ we know that it has an outgoing edge $vv'$ with weight $\geq 0$, and therefore thanks to Lemma~\ref{lem:technical-sup} we have $\supsum^N_\sigma(v) \geq w(vv') + \supsum^N_\sigma(v')$.
    \begin{itemize}
        \item If this edge is optimal, then this is an equality, so we set $\tau(v)=vv'$ and we are done.
        \item Otherwise, there is another optimal edge $vv''$, and therefore it satisfies 
        \[
            w(vv'') + \supsum^N_\sigma(v'') > w(vv') + \supsum^N_\sigma(v') \geq 0,
        \]
        hence by Lemma~\ref{lem:technical-sup} we have $\supsum^N_\sigma(v)=w(vv'') + \supsum^N_\sigma(v'')$ and we set $\tau(v)=vv''$.
    \end{itemize}
    Hence the claim is proved; we let $\tau$ denote such a Max strategy.

    Consider the unique play from a given vertex $v \in S$ which is consistent with both $\sigma$ and $\tau$; we should prove that its value is $\leq m + \phi_H(v)$.
    \begin{itemize}
        \item If the play remains forever in $S$. Note that the play is a lasso, i.e.~of the form
        \[
            v \re{} \dots \re{} u \re{} \dots \re{} u \re{} \dots \re{} u \re{} \dots.
        \]
        By the above claim, get that
        \[
            \supsum^N_\sigma(u) = \text{ weight of the cycle based on $u$ } + \supsum^N_\sigma(u).
        \]
        Since there are no zero cycles, this implies that $\supsum^N_\sigma(u) = +\infty$.
        However, since the play is consistent with $\sigma$ and does not visit a good escape, it must visit only edges with $\phi_H$-modified weight $\leq 0$, and therefore its value cannot be $\infty$.
        \item Otherwise, the play exits $S$, and since it is consistent with $\sigma$, it must exit via a good escape, say $uu'$.
        By the above claim we know that
        \[
            \begin{array}{lcl}
                \supsum^N_\sigma(v) &=& \text{ weight of prefix up to u } + w(uu') + \supsum_\sigma^N(u') \\
                &=& \text{ modified weight of prefix up to u } + \phi_H(v) - \phi_H(u) + w(uu') + \supsum^N(u') \\
                & \leq & \phi_H(v) - \phi_H(u) + w(uu') + \supsum^N(u') \leq \phi_H(v) +m,
            \end{array}
        \]
        since modified weights of the prefix are $\leq 0$ by definition of $\sigma$ and $\supsum_\sigma^N(u')=\supsum^N(u')$ holds because $\sigma$ is optimal from $F$.
    \end{itemize}
    This concludes the proof that for every $v \in S$ we have $\sup^N(v) \leq m + \phi_H(v)$.

    \vskip1em
    
    To get the other inequality, fix a vertex $u \in \Vmin \cap H^+$ which has a good escape as an outgoing edge $uu'$; we already know from Lemma~\ref{lem:main1}a that $\supsum^N(u)= w(uu') + \supsum^N(u')$ and by definition of $m$, this coincides with $\phi_H(u) +m$.
    Consider the game\footnote{We allow rational weights for the purpose of this proof.} $G'$ obtained from $G$ by adding the edge
    \[
        u \re{1/2 + \phi_H(u) - \phi_H(v)} v.
    \]
    Note that $G'$ also has no cycle of weight zero (this is the reason for using $1/2$ above and not $0$).
    Since the new edge has $\phi_H$-modified weight $1/2$, $H^+$ is still positively reduced by $\phi_H$ in $G'$, and therefore Lemma~\ref{lem:main1}a applies also to $G'$.
    Since in $G'$, the good escape $uu'$ in $G$ is still a good escape in $G'$ (we haven't added any escape), the lemma says that it is still an optimal edge for $\supsum^N$ in $G'$.
    Therefore the $\supsum^N$ values are the same in $G$ and $G'$ (since an optimal strategy does not use the new edge), and moreover
    \[
        \underbrace{\supsum^N(u)}_{\text{in } \Z} \leq 1/2 + \underbrace{\phi_H(u) - \phi_H(v) + \supsum^N(v)}_{\text{in } \Z},
    \]
    and therefore
    \[
        \supsum^N(u) \leq \phi_H(u) - \phi_H(v) + \supsum^N(v).
    \]
    This leads to the wanted bound:
    \[
        \supsum^N(v) \geq \supsum^N(u) + \phi_H(v) - \phi_H(u) = m + \phi_H(v),
    \]
    which concludes the proof of the first item.

    \vskip1em
    We now focus on the second item.
    The fact that there is an edge from $H^- \cap \Vmax$ to $F$ is given by Lemma~\ref{lem:main1}b.
    For the second statement, we first prove that for every vertex $v \in S$ we have $\supsum^N(v) \geq m+\phi_H(v)$.
    Consider a Max strategy $\tau$ which ensures that over $S$, a good escape is seen before any edge with $\phi_H$-modified weight $<0$ or any edge leaving from $H^-$ and from $F$, it is optimal for $\supsum^N$.
    Consider a vertex $v \in S$ and a play from $v$ consistent with $\tau$.
    If the play remains in $S$, then by definition of $\tau$ only edges with modified weight $\geq 0$ are seen, so the value is $\infty$ since there are no zero cycles.
    Otherwise, the play remains in $S$ until visiting a good escape $uu'$.
    Then, since $\tau$ is optimal from $u'$ we get 
    \[
        \begin{array}{lcl}
        &&\text{value of the play} \\&\geq& \text{weight of prefix up to u} + w(uu') + \supsum^N(u') \\
        & = & \underbrace{\text{modified weight of prefix up to u}}_{\geq 0} + \, \phi_H(v) \underbrace{- \phi_H(u) + w(uu') + \supsum^N(u')}_{\geq m} \\
        & \geq &\phi_H(v) +m.
        \end{array}
    \]
    
    There remains to prove that for every $v \in S$ we have $\sup^N(v) \leq m + \phi_H(v)$; for this inequality, we proceed just like in the previous item.
    Fix a vertex $u \in \Vmax \cap H^-$ which has a good escape as an outgoing edge $uu'$; we already know from Lemma~\ref{lem:main1}b that $\sup^N(u)=\phi_H(u)+m$.
    Consider the game $G'$ obtained from $G$ by adding the edge
    \[
        u \re{-1/2 + \phi_H(u) - \phi_H(v)} v.
    \]
    Note that $G'$ also has no cycle of weight zero.
    Since the new edge has $\phi_H$-modified weight $-1/2$, $H^-$ is still negatively reduced by $\phi_H$ in $G'$, and therefore Lemma~\ref{lem:main1}b applies also to $G'$.
    Since in $G'$, the optimal escape $uu'$ in $G$ is still an optimal escape in $G'$, the lemma says that it is still an optimal edge for $\supsum^N$ in $G'$.
    Therefore the $\supsum^N$ values are the same in $G$ and $G'$, and moreover
    \[
        \underbrace{\supsum^N(u)}_{\text{in } \Z} \geq -1/2 + \underbrace{\phi_H(u) - \phi_H(v) + \supsum^N(v)}_{\text{in } \Z},
    \]
    and therefore
    \[
        \supsum^N(u) \geq \phi_H(u) - \phi_H(v) + \supsum^N(v),
    \]
    which gives the wanted bound:
    \[
        \supsum^N(v) \leq \supsum^N(u) + \phi_H(v) - \phi_H(u) = m + \phi_H(v),
    \]
    which concludes the proof.
\end{proof}

We omit pseudocode for the computation of $S$ since it is similar to e.g.~the backtracking from the previous optimisation.

\subsection{Variant: choice between $\supsum^N$ or $\infsum^P$}

In line~\ref{line:choice}, we choose between computing $\supsum^N$ or $\infsum^P$ according to whether or not $|N| \leq |P|$.
This choice has no influence on the correctness or termination of the algorithm; any way of choosing between one or the other (including e.g.~always choosing $\supsum^N$) gives a valid algorithms.
One could imagine other choice policies, for instance the opposite one: choosing $\supsum^N$ over $\infsum^P$ whenever $|N| \geq |P|$.
Perhaps a meaningful choice, assuming the optimisation from Section~\ref{sec:optim1}, would be choosing $\supsum^N$ over $\infsum^P$ whenever $|S_N| \geq |S_P|$, where $S_N$ (resp. $S_P$) is the set of vertices from which Min (resp. Max) can ensure that $N$ (resp. $P$) is seen before any edge with weight $>0$ (resp. $<0$).

\subsection{Variant: remembering the potentials}\label{sec:variant-remember-potentials}

In the execution of Algorithm~\ref{alg:main}, just before the recursive call is made at line~\ref{line:main-recurse} on the smaller game $H$, we could apply an arbitrary potential reduction to $H$ without altering the correction or termination of the algorithm.
A meaningful choice would be applying the potential reduction associated to the potential $\phi_{H_{\pred}}$ which was computed in the previous iteration of the while loop.
At the level of intuition, this means that we do not discard the work that was performed in the previous iteration, and instead re-use it.
Intuitively, if $H_\pred$ (the game $H$ from the previous iteration of the while loop) and $H$ are similar (i.e.~only one vertex, or a few vertices if the above optimisation is used, has been removed), then we should expect that the potential $\phi_{H_\pred}$ is close to making $H$ reduced, so performing the potential update should reduce the runtime.
(Note that with this variation, we overcome the principal weakness of Zielonka's algorithm, where much of the information is discarded after each recursive call.)

Formally, we add a new variable, say $\phi_\pred$, containing a potential, initialised to $0$, and we add $H \gets \FReduceGame(H,\phi_\pred)$ just before line~\ref{line:main-recurse} and $\phi_\pred \gets \phi_H$ just after line~\ref{line:main-recurse}.

\subsection{Variant: lazy version}

In this last (drastic) variant, we propose to end the computation of $\supsum^N$ as soon as we detect that a vertex would leave $N_G$.
To achieve this, whenever we add a vertex $v'$ to $F$ in lines~\ref{line:add-v-H+},~\ref{line:add-v-H-} or inside the subprocedure from line~\ref{line:backtrack-all-paths}, we inspect its predecessors $v$ which belong to $N$ and determine whether the new weight of $vv'$ (i.e.~$w(vv')+\sup^N(v')$) is so that $v$ would exit $N$ (this occurs if $v$ is a Max vertex and $w(vv')+\sup^N(v')\geq 0$ or $v$ is a Min vertex, $w(vv')+\sup^N(v')\geq 0$ and there is no other outgoing edge from $v$ with weight $<0$.)
In this case, we fast-forward to line~\ref{line:potential-reduction}, i.e.~we immediately perform the potential update (corresponding to the current potential $\supsum^N$, which at this stage corresponds to the $\supsum^N$-values over $F$ and has value $0$ elsewhere) and the recursive call from the next line.
Termination is still guaranteed because the size of $N$ decreases.

\section{Discussion}\label{sec:futurework}

We have introduced a new algorithm for solving mean-payoff games.
Compared to existing ones, we believe that this algorithm has several novelties: it is recursive, based on backtracking (in parity-games terminology, one may call it ``attractor-based''), and symmetric.
One may be tempted to call it a mean-payoff analogue of Zielonka's algorithm for parity games~\cite{Zielonka98} which shares these features; however ours is more involved and differs from Zielonka's when executed on a parity game transformed into a mean-payoff game.

We believe that our algorithm is a candidate for having subexponential runtime.
This may require adding the two optimisations as well as using one (or several) of the variants from Section~\ref{sec:optimisations-and-variant}.
So far, we haven't been able to analyse the algorithm in a meaningful way, so we leave this as an open question.

A quick implementation also seems to reveal that our algorithm may be relevant for practical use; we also leave this interrogation to future work (and more qualified programmers).

\bibliography{bib}

@article{Zielonka98,
  author       = {Wieslaw Zielonka},
  title        = {Infinite Games on Finitely Coloured Graphs with Applications to Automata
                  on Infinite Trees},
  journal      = {Theoretical Computer Science},
  volume       = {200},
  number       = {1-2},
  pages        = {135--183},
  year         = {1998},
  bibsource    = {dblp computer science bibliography, https://dblp.org}
}

@inproceedings{OhlmannGKK,
  author       = {Pierre Ohlmann},
  title        = {The {GKK} Algorithm is the Fastest over Simple Mean-Payoff Games},
  booktitle    = {{CSR} 2022},
  series       = {LNCS},
  volume       = {13296},
  pages        = {269--288},
  publisher    = {Springer},
  year         = {2022},
  bibsource    = {dblp computer science bibliography, https://dblp.org}
}

@article{BV07,
  author       = {Henrik Bj{\"{o}}rklund and
                  Sergei G. Vorobyov},
  title        = {A combinatorial strongly subexponential strategy improvement algorithm
                  for mean payoff games},
  journal      = {Discret. Appl. Math.},
  volume       = {155},
  number       = {2},
  pages        = {210--229},
  year         = {2007},
  bibsource    = {dblp computer science bibliography, https://dblp.org}
}

@article{GKK88,
  author    = {V. A. Gurvich and A. V. Karzanov and L. G. Khachiyan},
  title     = {Cyclic games and an algorithm to find minimax cycle means in directed graphs},
  journal   = {USSR Computational Mathematics and Mathematical Physics},
  volume    = {28},
  pages     = {85--91},
  year      = {1988}
}

@article{ZP96,
  author    = {Uri Zwick and Mike Paterson},
  title     = {The complexity of mean payoff games on graphs},
  journal   = {Theoretical Computer Science},
  volume    = {158},
  number    = {1-2},
  pages     = {343--359},
  year      = {1996}
}

@inproceedings{BFLMS08,
  author    = {Patricia Bouyer and Uli Fahrenberg and Kim G. Larsen and Nicolas Markey and Jiří Srba},
  title     = {Infinite Runs in Weighted Timed Automata with Energy Constraints},
  booktitle = {FORMATS’08},
  series    = {LNCS},
  volume    = {5215},
  pages     = {33--47},
  publisher = {Springer‑Verlag},
  year      = {2008},
  month     = sep,
  
}

@inproceedings{CAHS03,
  author    = {Arindam Chakrabarti and Luca de Alfaro and Thomas A. Henzinger and Mariëlle Stoelinga},
  title     = {Resource Interfaces},
  booktitle = {EMSOFT 2003},
  series    = {LNCS},
  volume    = {2855},
  pages     = {117--133},
  publisher = {Springer, Berlin, Heidelberg},
  year      = {2003},
  
}

@inproceedings{LS24,
  author       = {Bruno Loff and
                  Mateusz Skomra},
  title        = {Smoothed Analysis of Deterministic Discounted and Mean-Payoff Games},
  booktitle    = {{ICALP} 2024},
  series       = {LIPIcs},
  volume       = {297},
  pages        = {147:1--147:16},
  publisher    = {Schloss Dagstuhl - Leibniz-Zentrum f{\"{u}}r Informatik},
  year         = {2024},
  
}

@inproceedings{DKZ19,
  author       = {Dani Dorfman and
                  Haim Kaplan and
                  Uri Zwick},
  title        = {A Faster Deterministic Exponential Time Algorithm for Energy Games
                  and Mean Payoff Games},
  booktitle    = {{ICALP} 2019},
  series       = {LIPIcs},
  volume       = {132},
  pages        = {114:1--114:14},
  publisher    = {Schloss Dagstuhl - Leibniz-Zentrum f{\"{u}}r Informatik},
  year         = {2019},
  
}

@inproceedings{CCO25,
  author       = {Micha{\"{e}}l Cadilhac and
                  Antonio Casares and
                  Pierre Ohlmann},
  title        = {Fast value iteration: {A} uniform approach to efficient algorithms
                  for energy games},
  booktitle    = {{TACAS} 2025},
  series       = {LNCS},
  volume       = {15697},
  pages        = {323--342},
  publisher    = {Springer},
  year         = {2025},
  bibsource    = {dblp computer science bibliography, https://dblp.org}
}

@inproceedings{Schewe08,
  author       = {Sven Schewe},
  title        = {An Optimal Strategy Improvement Algorithm for Solving Parity and Payoff
                  Games},
  booktitle    = {{CSL} 2008},
  series       = {LNCS},
  volume       = {5213},
  pages        = {369--384},
  publisher    = {Springer},
  year         = {2008},
  bibsource    = {dblp computer science bibliography, https://dblp.org}
}

@article{EM73,
author = {Andrzej Ehrenfeucht and Jan Mycielski},
title = {Positional games over a graph},
journal = {Notices of the American Mathematical Society},
volume={20},
year = {1973},
pages={A-334}
}

@article{EM79,
  author    = {Andrzej Ehrenfeucht and Jan Mycielski},
  title     = {Positional strategies for mean payoff games},
  journal   = {International Journal of Game Theory},
  volume    = {109},
  number    = {8},
  pages     = {109--113},
  year      = {1979},
 
  
}

@article{BCDGR11,
  author       = {Lubos Brim and
                  Jakub Chaloupka and
                  Laurent Doyen and
                  Raffaella Gentilini and
                  Jean{-}Fran{\c{c}}ois Raskin},
  title        = {Faster algorithms for mean-payoff games},
  journal      = {Formal Methods Syst. Des.},
  volume       = {38},
  number       = {2},
  pages        = {97--118},
  year         = {2011},
  
}

\appendix

\end{document}